\documentclass[12pt]{article}
\usepackage[margin=1in]{geometry} 
\usepackage{amsmath,amsthm,amssymb,todonotes}
\usepackage{multirow}
\usepackage{times}
\usepackage{listings}
\usepackage{graphicx}
\usepackage{enumerate}
\usepackage{mathrsfs}
\usepackage[colorlinks,
linkcolor=blue,
anchorcolor=blue,
citecolor=blue
]{hyperref}
\usepackage{dsfont}
\usepackage{float}
\usepackage[title]{appendix}
\usepackage{booktabs}
\usepackage{caption}
\usepackage{tikz}
\usetikzlibrary{decorations.pathreplacing}

\usepackage{chngcntr}

\usepackage{natbib}

\usepackage{setspace}
\newcommand{\R}{\mathbb{R}}

\makeatletter

\newcommand{\Rmnum}[1]{\expandafter\@slowromancap\romannumeral #1@}
\makeatother
\allowdisplaybreaks

\newtheorem{theorem}{Theorem}[section]
\newtheorem{proposition}[theorem]{Proposition}
\newtheorem{lemma}[theorem]{Lemma}
\newtheorem{definition}[theorem]{Definition}
\newtheorem{corollary}[theorem]{Corollary}

\numberwithin{equation}{section}

\title{A Concavification Approach to Ambiguous Persuasion}
\author{Xiaoyu Cheng\footnote{
		Department of Managerial Economics and Decision Sciences, Kellogg School of Management, Northwestern University, Evanston, IL, USA. E-mail: xiaoyu.cheng@kellogg.northwestern.edu}}
\begin{document}
	\maketitle
	
	\begin{abstract}
		This note shows that the value of ambiguous persuasion characterized in \citet*{beauchene2019ambiguous} can be given by a concavification program as in Bayesian persuasion \citep*{kamenica2011bayesian}. In addition, it implies that an ambiguous persuasion game can be equivalently formalized as a Bayesian persuasion game by distorting  the utility functions. This result is obtained under a novel construction of ambiguous persuasion. 
	\end{abstract}
	
	\section{Introduction}
	\citet*{beauchene2019ambiguous} (BLL henceforth) develops a theory of ambiguous persuasion where the sender is allowed to commit to an ambiguous information structure. Under the assumption that both the sender and receiver are Maxmin Expected Utility (MEU) decision makers \citep*{gilboa1989maxmin}  and the receiver applies full Bayesian updating \citep*{pires2002rule}, they characterize the value of ambiguous persuasion and show that it can be strictly higher than the value of Bayesian persuasion. 
	
	Their characterization is given by a maximal projection of the concave closure of a value function which is different from the standard concavification in Bayesian persuasion. The present note shows that the value of ambiguous persuasion characterized in their paper in fact can be given by a standard concavification program. It further implies that, the ambiguous persuasion game is equivalent to a Bayesian persuasion game where the players' utility functions are distorted. 
	
	This result is built on a novel construction of ambiguous persuasion different from BLL. In Bayesian persuasion, \cite{kamenica2011bayesian} show that the sender's objective is to choose among \textit{distributions over posteriors}. For  ambiguous persuasion, a natural generalization adopted by BLL is to let the sender choose among \textit{sets of distributions over posteriors}. In the present note, I show it is equivalent for the sender to choose from \textit{distributions over sets of posteriors}. 
	
	Under this construction, the \textit{Bayes Plausibility} condition is generalized to a condition called \textit{Verifiable Bayes Plausibility}. By definition, a distribution over sets of posteriors can be verified to be Bayes plausible if one can select a posterior from each set of posteriors such that the distribution over the selected posteriors satisfies the Bayes plausibility. 
	
	While not all ambiguous information structure give rise to a distribution over sets of posteriors, I show that any distribution over sets of posteriors that is Verifiably Bayes Plausible (VBP) can be induced by an ambiguous information structure. Moreover, the sender's optimal value from choosing among VBP distributions coincides with the optimal value from ambiguous persuasion. Therefore, Verifiable Bayes Plausibility is exactly the counterpart of Bayes Plausibility in ambiguous persuasion. The concavification result is then derived from this observation.\\
	
	The remainder of this note is organized as follows. Section \ref{result} defines the ambiguous persuasion game and states the main result. Section \ref{construction} describes the new construction and proves the main result.

	\section{The Main Result}\label{result}
	\subsection{Setup}
	Let $\Omega$ denote a finite set of states. For any set $X$, let $\Delta(X)$ be the space of probability distributions over $X$ endowed with the weak topology. The sender and receiver have a common \textit{prior} $p_{0} \in \Delta(\Omega)$ with full support. 
	
	The receiver's feasible actions are given by a compact set $A$. The sender and receiver's payoff depend on the receiver's action and state, given by continuous functions $u_{S}$ and $u_{R}$, respectively. Namely, $u_{S}(a, \omega)$ and $u_{R}(a,\omega)$ are the payoffs of the players when the receiver's action is $a \in A$ and the realized state is $\omega \in \Omega$. Let $f \in \Delta(A)$ denote a generic strategy of the receiver. 
	
	Let $\mathcal{K}(\Omega)$ denote the set of all closed and convex subsets of $\Delta(\Omega)$. Fix any $P \in \mathcal{K}(\Omega)$, let $\bar{f}(P)$ denote the set of strategies that is optimal to the receiver when his belief is given by the set $P$. Namely, 
	\begin{equation*}
	\bar{f}(P) \equiv \arg\max\limits_{f \in \Delta(A)} \min\limits_{p \in P} \sum\limits_{a,\omega} \hat{f}(P)(a) u_{R}(a,\omega)p(\omega)
	\end{equation*}
	
	When $\bar{f}(P)$ is not a singleton, for every $p \in P$, define $\hat{f}_{p}(P) \in \bar{f}(P)$ to be the sender's most preferred strategy evaluated at belief $p$. That is, 
	\begin{equation*}
	\sum\limits_{a,\omega} \hat{f}_{p} (P)(a) u_{S}(a,\omega)p(\omega) \geq \sum\limits_{a,\omega} f(a) u_{S}(a,\omega)p(\omega)  \quad \forall f \in \bar{f}(P)
	\end{equation*}
	
	The sender's preferred action also depends on the posterior where she evaluates it is a new feature of ambiguous persuasion. It will become clear in the next section that the sender can always induce the receiver to form a posterior belief $P$ meanwhile she evaluates the receiver's actions at the posterior $p \in P$. 
	
	\subsection{The Concavification}
	Let $v: \Delta(\Omega) \rightarrow \R$ be the sender's value function defined by the following: 
	\begin{equation*}
		v(p) =  \max\limits_{P \in \mathcal{K}(\Omega):  p \in P} \sum\limits_{a,\omega}\hat{f}_{p}(P)(a)u_{R}(a,\omega)p(\omega)
	\end{equation*}
	Namely, the sender's value at a distribution $p$ is given by maximizing over all possible sets of distributions containing it. Given the receiver takes the sender-preferred action evaluating at $p$, $v(p)$ is clearly upper semi-continuous. Let $\hat{v}(p)$ denote the concave closure of $v(p)$: 
	\begin{equation*}
		\hat{v}(p) \equiv \inf \{H(p) | H: \Delta(\Omega) \rightarrow \R, H \geq v, H \text{ is affine and continuous}\}
	\end{equation*}
	
	The main result of this note claims that the concavification of this value function characterizes the value of optimal ambiguous persuasion. 
	
	\begin{theorem}\label{concave}
		The value of optimal ambiguous persuasion at prior $p_{0}$ is equal to $\hat{v}(p_{0})$. 
	\end{theorem}

	The proof of this theorem relies on a novel construction of the ambiguous persuasion game which will be provided in the next section. Theorem \ref{concave} implies that the ambiguous persuasion game is equivalent to a Bayesian persuasion game with $v(p)$ being the sender's corresponding value function. Therefore, the optimal value can be solved in the same way as solving a Bayesian persuasion problem. 
	
	\begin{corollary}\label{bayesian}
		The value of optimal ambiguous persuasion is given by the following: 
		\begin{align*}
			&\max\limits_{\tau \in \Delta(\Delta(\Omega))}	\mathbb{E}_{\tau}[v(p)] \\
			&\text{s. t. } \sum\limits_{p\in supp(\tau)} \tau(p) p = p_{0}
		\end{align*}
	\end{corollary}

	\section{The Construction} \label{construction}
	\subsection{Information Structure} 
	Fix a finite message space $M$, the sender can commit to a probabilistic device $\pi$ which is a mapping from $\Omega$ to probability distributions over $M$. Let $\pi(m|\omega) $ denote the probability of sending message $m$ in state $\omega$ under the device $\pi$. Moreover, let $\tau_{\pi}(m) \equiv \sum_{\omega} p_{0}(\omega) \pi(m|\omega)$ denote the overall probability of sending message $m$ under device $\pi$. An ambiguous device $\Pi$ is a closed and convex set of probabilistic devices with common support\footnote{For any $\pi, \pi' \in \Pi$, $\tau_{\pi}(m)>0$ if and only if $\tau_{\pi'}(m)>0$. }. 
	
	Fix an ambiguous device $\Pi$, each realized message $m$ induces a \textit{probability-possibly} set $P_{m}$ defined as the set of posteriors from applying Bayes' rule with respect to every probabilistic device in $\Pi$. Namely, 
	\begin{equation*}
		P_{m}^{\Pi} \equiv  \left\{p_{m}^{\pi} \in \Delta(\Omega): p_{m}^{\pi}(\cdot) = \frac{\pi(m|\cdot)p_{0}(\cdot)}{\sum\limits_{\omega' \in \Omega} \pi(m|\omega')p_{0}(\omega')}, \pi \in \Pi \right\} 
	\end{equation*}
	Notice that $P_{m}^{\Pi} \in \mathcal{K}(\Omega)$. When the receiver applies full Bayesian updating, the probability-possibility set is exactly the sets of posterior beliefs after observing the realized message $m$. 
	
	\subsection{Verifiable Bayes Plausibility}
	For any $P, Q \in \mathcal{K}(\Omega)$, a distribution over closed and convex sets of distributions is denoted by $\mu \in \Delta(\mathcal{K}(\Omega))$. If a distribution $\mu$ is induced by an ambiguous device, then it needs to satisfy the following condition. 

\begin{definition}
	A distribution $\mu$ is \textbf{Verifiably Bayes Plausible} if there exists a selection function $\varphi: \mathcal{K}(\Omega)\rightarrow \Delta(\Omega)$ with $\varphi(P) \in P$ such that 
	\begin{equation*}
		\sum\limits_{P\in supp(\mu)} \mu(P) \varphi(P) = p_{0}
	\end{equation*}
\end{definition}

	A distribution $\mu$ is verifiably Bayes plausible if Bayes plausibility can be verified by selecting a posterior from each set of posteriors in its support. Such a selection is then called the \textit{verifying selection}. A verifiably Bayes plausible (VBP) distribution $\mu$ may have multiple verifying selections. Let $\Phi_{\mu}$ denote the set of all verifying selections of $\mu$, i.e. 

\begin{equation*}
	\Phi_{\mu} = \left\{\varphi(\cdot): \sum\limits_{P \in supp(P)} \mu(P)\varphi(P) = p_{0} \right\}
\end{equation*}

Fix a VBP $\mu$ and one of its verifying selection $\varphi(\cdot)$, $\varphi(P)$ is then a \textit{verifying posterior} of $P$. Let $P_{\mu}$ denote the set of all verifying posteriors of $P$ given distribution $\mu$, i.e.

\begin{equation*}
	P_{\mu} \equiv \{p \in P: p = \varphi(P) \text{ for some } \varphi \in \Phi_{\mu} \}
\end{equation*}

The first observation is that $P_{\mu}$ is also a closed and convex subset of $\Delta(\Omega)$, which is an immediate consequence of the sets in the support of $\mu$ being closed and convex. 
\begin{proposition}\label{vp}
	For any verifiably Bayes plausible $\mu$ and $P \in supp(\mu)$, $P_{\mu}$ is a closed and convex set. 
\end{proposition}
\begin{proof}
	First show convexity. Given any VBP $\mu$ and $P \in supp(\mu)$, consider the set $P_{\mu}$. Suppose $p_{1}, p_{2} \in P_{\mu}$ such that $p_{1} = \varphi_{1}(P)$ and $p_{2} = \varphi_{2}(P)$. Then it implies 
	\begin{align*}
		\mu(P)p_{1} + \sum\limits_{P' \in supp(\mu)\backslash P} \mu(P') \varphi_{1}(P') = p_{0}\\
		\mu(P)p_{2} + \sum\limits_{P' \in supp(\mu)\backslash P} \mu(P') \varphi_{2}(P') = p_{0}
	\end{align*}
	For any $\lambda \in [0,1]$, let $\varphi_{\lambda}(\cdot) = \lambda \varphi_{1}(\cdot) + (1-\lambda)\varphi_{2}(\cdot)$. Then combine the two equations above implies that $\varphi_{\lambda}$ is also a verifying selection of $\mu$:
	\begin{equation*}
		\mu(P) \lambda [p_{1}+(1-\lambda)p_{2}] + \sum\limits_{P' \in supp(\mu)\backslash P} [\lambda \varphi_{1}(P') + (1-\lambda) \varphi_{2}(P')] = p_{0}
	\end{equation*} 
	Therefore, $\varphi_{\lambda}(P) = \lambda p_{1} + (1-\lambda)p_{2}$ is also contained in $P_{\mu}$. Thus $P_{\mu}$ is convex. 
	
	For closeness, similarly consider a sequence of verifying posteriors $\{p_{n}\}_{n=1,2,\cdots}$ of $P_{\mu}$ with corresponding verifying selection $\{\varphi_{n}\}_{n=1,2,\cdots}$. Suppose $\lim_{n \rightarrow \infty}p_{n} = p$, it implies that $\lim_{n \rightarrow \infty}\varphi_{n}(P) = p \equiv \varphi(P)$. The last term is well defined because $P$ is  closed. Then it is easy to verify that such $\varphi$ is a verifying selection of $\mu$. Therefore, $p \in P_{\mu}$. $\square$. 
\end{proof}

\subsection{Two-step Construction of Ambiguous Devices}
Clearly, not all ambiguous devices can induce a distribution over set of posteriors. What I will show next is that every VBP distribution can be induced by some ambiguous device. In particular, any such ambiguous device can be constructed in a two-step manner. 

For the first step, say that an ambiguous device $\Pi$ is \textit{simple} if for all $\pi, \pi' \in \Pi$, $\tau_{\pi}(m) = \tau_{\pi'}(m)$. Namely, a simple ambiguous device consists of probabilistic devices that generate the same message with the same overall probability. As a result, a simple ambiguous device $\Pi$ generates the probability-possibility set $P^{\Pi}_{m}$ with probability $\tau_{\pi}(m)$. In other words, it induces a distribution $\mu$ with $\mu(P_{m}) = \tau_{\pi}(m)$. Moreover, notice that the induced distribution is always \textit{fully-verified}:

\begin{definition}
	A verifiably Bayes plausible distribution $\mu$ is \textbf{fully-verified} if for all $P \in supp(\mu)$, $P_{\mu } = P$. 
\end{definition}

\begin{lemma}\label{simple}
	A simple ambiguous device induces a fully-verified verifiably Bayes plausible distribution. A fully-verified verifiably Bayes plausible distribution can be induced by a simple ambiguous device. 
\end{lemma}

The proof of Lemma \ref{simple} is straightforward thus skipped here. Clearly, not all VBP distributions are fully-verified. Thus simple ambiguous devices alone cannot induce all VBP distributions. However, it suffices to consider ambiguous devices that \textit{dilates} a simple ambiguous device. 

Fix a simple ambiguous device $\Pi$ with the message space $M$. Consider a dilating message space $\tilde{M}$ and a function $g$ maps from $\Omega \times M$ to $\Delta(\tilde{M})$. Moreover, assume that the support of $g(\cdot|m, \omega)$ is $\{\tilde{m}^{i}\}_{i=1,2,\cdots} \subseteq \tilde{M}$ for all $\omega \in \Omega$. Therefore, $\tilde{M}$ can be thought of as a refinement of $M$. 

Let $g \circ \pi$ denote the composition of $g$ and $\pi$, which is a mapping from $\Omega$ to $\Delta(\tilde{M})$:
\begin{equation*}
	(g \circ \pi)(\tilde{m}^{i}|\omega) = \sum\limits_{m'\in M} g(\tilde{m}^{i}|m',\omega') \pi(m'|\omega) = g(\tilde{m}^{i}|m,\omega) \pi(m|\omega)
\end{equation*}
where the second equality follows from $g(\tilde{m}^{i}|m', \omega) > 0$ only when $m' = m$. Therefore, $g \circ \pi$ is a probabilistic device with $\tilde{M}$ being the message space. The posterior corresponding to message $\tilde{m}_{i}$ is then  given by
\begin{equation*}
	q_{\tilde{m}^{i}}^{g\circ \pi}(\omega) = \frac{g(\tilde{m}^{i}|m,\omega) \pi(m|\omega)p_{0}(\omega)}{\sum\limits_{\omega' \in \Omega} g(\tilde{m}^{i}|m,\omega') \pi(m|\omega')p_{0}(\omega')    }
\end{equation*}

Let $\tau_{g\circ \pi} (\tilde{m}^{i}|m)$ denote the overall probability of message $\tilde{m}^{i}$, i.e.
\begin{equation*}
	\tau_{g\circ \pi} (\tilde{m}^{i}|m) = \frac{\sum\limits_{\omega \in \Omega} g(\tilde{m}^{i}|m,\omega) \pi(m|\omega)p_{0}(\omega)}{\sum\limits_{\omega \in \Omega}\pi(m|\omega)p_{0}(\omega) }
\end{equation*}

Then it is easy to verify that the following conditional Bayes plausibility holds:
\begin{equation*}
	\sum\limits_{i} \tau_{g\circ \pi} (\tilde{m}^{i}|m) \cdot 	q_{\tilde{m}^{i}}^{g\circ \pi}(\cdot) = q_{m}^{\pi}(\cdot)
\end{equation*}
Namely, the posteriors generated from refining the message $m$ to $\{\tilde{m}^{i}\}_{i=1,2,\cdots}$ also needs to average back to the posterior given $m$. Then let $P^{g\circ \pi}_{\tilde{m}}$ denote the convex hull of the posteriors $\{ q_{\tilde{m}^{i}}^{g\circ \pi}(\cdot)\}_{i=1,2,\cdots}$ it must be the case that $q_{m}^{\pi} \in P^{g\circ \pi}_{\tilde{m}}$. On the other hand, it also implies that given a posterior $q \in \Delta(\Omega)$ any set $P \subseteq(\Omega)$ containing it can be induced by some function $g$ in this way. 

Next, to let every message $\tilde{m}^{i}$ be able to induce the whole set $P^{g\circ \pi}_{\tilde{m}}$, it suffices to consider permutations of $g$. For example, let $g'$ permute the label of $\tilde{m}^{i}$ and $\tilde{m}^{j}$ such that 
\begin{align*}
	&g' (\tilde{m}^{i}|m,\omega) = g(\tilde{m}^{j}|m,\omega)  \\
	&g' (\tilde{m}^{j}|m,\omega) = g(\tilde{m}^{i}|m,\omega)  
\end{align*} 
for all $\omega \in \Omega$. Then fix $\pi$ and suppose the set $G \equiv co(\{g, g'\})$ is used to generate messages in $\tilde{M}$ in an ambiguous way. As a result, the probability-possibility set $P_{\tilde{m}^{i}}$ and $P_{\tilde{m}^{j}}$ will coincide and equal the convex hull of $p_{\tilde{m}^{i}}^{g\circ \pi}$ and $p_{\tilde{m}^{j}}^{g\circ \pi}$. Analogously, the whole set of posteriors $P^{g\circ \pi}_{\tilde{m}}$ can be generated at any message $\tilde{m}^{i}$ by constructing a set of all possible permutations of $g$. 

Finally, for a simple ambiguous device $\Pi$ inducing the set $P^{\Pi}_{m}$ at message $m$. The above construction can be applied to every $\pi \in \Pi$ such that every construction leads to the same set of posteriors $P_{\tilde{m}}$. Notice that $P^{\Pi}_{m} \subseteq P_{\tilde{m}}$ by construction. Effectively, the messages $\{ \tilde{m}^{i}\}_{i=1,2,\cdots}$ dilate the set $P^{\Pi}_{m}$. 

To summarize, any VBP distribution $\mu$ can be induced by an ambiguous device constructed in the following two steps: 
\begin{enumerate}[(i)]
	\item Find a fully-verified VBP distribution $\mu'$ that satisfies for each $P \in supp(\mu)$ there exists $P' \in supp(\mu')$ such that $P' \subseteq P$ and $\mu'(P') = \mu(P)$. Then construct a simple ambiguous device $\Pi$ induces $\mu'$. 
	\item For each $\pi \in \Pi$, identify the function $g$ dilating the posteriors induced by $\pi$ to the sets $P$ in the support of $\mu$. Then construct probabilistic devices using $\pi$ and all possible permutations of $g$. 
\end{enumerate}
Let $G \circ \Pi$ denote the ambiguous device constructed in this way. Clearly, it contains only probabilistic devices in the form of $g \circ \pi$. Therefore, the following proposition can be proved directly from this construction. 
\begin{proposition}\label{vbp}
	Any verifiably Bayes plausible distribution can be induced by an ambiguous device. 
\end{proposition}

Proposition \ref{vbp} suggests that VBP distributions are relevant objectives to search for the optimal ambiguous information structure. However, since the ambiguous devices constructed in this way are special cases of all ambiguous devices. It is not automatically true that optimal ambiguous persuasion will be given in this form. 

\subsection{Value of Persuasion}
Recall $\hat{f}_{p}(P)$ denotes the receiver's optimal strategy at a set of posteriors $P$ gives the sender most expected payoff according to the belief $p \in P$. The explicit dependence on $p$ is unique to the current construction, since a VBP distribution $\mu$ can be induced by multiple ambiguous devices. Thus, fix a VBP distribution $\mu$, the sender is able to choose her most preferred ambiguous device. Then under this device, the receiver's strategy will be evaluated at the belief $p$. Formally, for any VBP distribution $\mu$ and a verifying selection $\varphi$, the sender's ex-ante payoff from the corresponding ambiguous device is given by: 
\begin{equation*}
	V_{S}(\mu, \varphi) = \sum\limits_{P \in supp(\mu)} \mu(P) \sum\limits_{a, \omega} \hat{f}_{\varphi(P)}(P)(a)u_{S}(a, \omega)\varphi(P)(\omega)
\end{equation*}
The following result establishes the equivalence of value between the ambiguous devices and VBP distributions with some verifying selection. 

\begin{proposition}\label{equivalence}
	The followings are equivalent: 
	\begin{enumerate}[(i)]
		\item There exists an ambiguous device with value $v^{*}$. 
		\item There exists a verifiably Bayes plausible distribution $\mu$ and a verifying selection $\varphi$ such that $V_{S}(\mu, \varphi) = v^{*}$. 
	\end{enumerate}
\end{proposition}
\begin{proof}
	The direction $(ii) \Rightarrow (i)$ is immediate given our two-step construction of ambiguous devices. For the other direction $(i) \Rightarrow (ii)$ , notice that 
	\begin{align*}
		v^{*} &= \min\limits_{\pi \in \Pi} \sum\limits_{\omega \in \Omega} p_{0}(\omega) \sum\limits_{m, a} \hat{f}_{q_{m}^{\pi}}(P_{m}^{\Pi})(a) \pi(m|\omega) u_{S}(a, \omega)\\
		& = \min\limits_{\pi \in \Pi}  \sum\limits_{m,a}  \hat{f}_{q_{m}^{\pi}}(P_{m}^{\Pi})(a) \left[ \sum\limits_{\omega \in \Omega} p_{0}(\omega)\pi(m|\omega) \right] \frac{\sum\limits_{\omega \in \Omega} p_{0}(\omega)\pi(m|\omega)u_{S}(a, \omega)}{\sum\limits_{\omega \in \Omega} p_{0}(\omega)\pi(m|\omega)} \\
		& = \min\limits_{\pi \in \Pi} \sum\limits_{m} \tau_{\pi}(m) \sum\limits_{a,\omega}\hat{f}_{q_{m}^{\pi}}(P_{m}^{\Pi})(a) u_{S}(a,\omega)q_{m}^{\pi}(\omega) \\
		& = \sum\limits_{m \in M} \tau_{\underline{\pi}}(m) \sum\limits_{a,\omega}\hat{f}_{q_{m}^{\underline{\pi}}}(P_{m}^{\Pi})(a) u_{S}(a,\omega)q_{m}^{\underline{\pi}}(\omega) 
	\end{align*}
	As $q^{\underline{\pi}}_{m} \in P^{\Pi}_{m}$ for all $m$ and $\sum_{m } \tau_{\underline{\pi}}(m) q^{\underline{\pi}}_{m} = p_{0}$. The distribution given by $\mu(P^{\Pi}_{m}) = \tau_{\underline{\pi}}(m)$ is verifiably Bayes plausible with verifying selection $\varphi(P_{m}) = q^{\underline{\pi}}_{m}$. Therefore, one has $V_{S}(\mu, \varphi) = v^{*}$. 
\end{proof}

\subsection{Characterizing Value by Concavification}
Notice that, fix any VBP distribution $\mu$, the sender is able to achieve the maximum value over all verifying selections. Let $\bar{V}_{S}(\mu)$ denote this value:  
\begin{equation*}
	\overline{V}_{S}(\mu) = \max\limits_{\varphi \in \Phi_{\mu}} \sum\limits_{P \in supp(\mu)} \mu(P) \sum\limits_{a, \omega} \hat{f}_{\varphi(P)}(P)(a)u_{S}(\hat{f}_{\varphi(P)}(P), \omega)\varphi(P)(\omega)
\end{equation*}
The maximum exists\footnote{$P_{\mu}$ is closed and convex implies $\Phi_{\mu}$ is also a closed and convex set of functions.} and it can be achieved by designing a Bayesian device in the first step that induces the maximizing verifying posteriors. 

As a result, the sender's optimal ambiguous persuasion can be simplified to the following program.  
\begin{corollary}
	The value of optimal ambiguous persuasion is given by the following: 
	\begin{align*}
		&\max\limits_{\mu \in \Delta(\mathbf{P})}	\overline{V}_{S}(\mu)\\
		&\text{s. t. } \mu\text{ being verifiably Bayes plausible.}
	\end{align*}
\end{corollary}

Now the main result, Theorem \ref{concave} can be proved as in the following. 

\begin{proof}[Proof of Theorem \ref{concave}]
	Let $V^{*}(p_{0})$ denote the value of optimal ambiguous persuasion solved from the following program: 
	\begin{align*}
		V^{*}(p_{0}) = 	&\max\limits_{\mu \in \Delta(\mathbf{P})}	\overline{V}_{S}(\mu) \\
		&\text{s. t. } \mu\text{ being verifiably Bayes plausible.}
	\end{align*}
	
	First of all, I show that $\hat{v}(p_{0})$ can be achieved at a VBP distribution $\mu$ with a verifying selection $\varphi$. 
	
	Let $\tau\in \Delta(\Delta(\Omega))$ be the distribution inducing the value $\hat{v}(p_{0})$, thus it satisfies Bayes plausibility. For each $p \in supp(\tau)$, let $P^{*}$ denote the probability-possibility set inducing the value $v(p)$. Then the distribution $\mu$ with $\mu(P^{*}) = \tau(p)$ and verifying selection $\varphi(P^{*}) = p$ clearly inducing the value $\hat{v}(p_{0})$. That is, $V_{S}(\mu, \varphi) = \hat{v}(p_{0})$. Therefore, 
	\begin{equation*}
		V^{*}(p_{0})  = \max\limits_{\mu \in \Delta(\mathbf{P})}	\overline{V}_{S}(\mu) \geq V_{S}(\mu, \varphi) = \hat{v}(p_{0})
	\end{equation*}
	
	Next, I am going to show that it cannot be the case $V^{*}(p_{0}) > \hat{v}(p_{0})$. Towards a contradiction, suppose the ex-ante value from some VBP distribution $\mu$ with verifying selection $\varphi$ is strictly higher than $\bar{v}(p_{0})$. Because 
	\begin{equation*}
		V_{S}(\mu, \varphi) = \sum\limits_{P \in supp(\mu)} \mu(P) \sum\limits_{a, \omega} \hat{f}_{\varphi(P)}(P)(a)u_{S}(\hat{f}_{\varphi(P)}(P), \omega)\varphi(P)(\omega)
	\end{equation*}
	is in the form of an expectation. It therefore implies that there must exist $P' \in supp(\mu)$ such that 
	\begin{equation*}
		\sum\limits_{a, \omega} \hat{f}_{\varphi(P)}(P)(a)u_{S}(\hat{f}_{\varphi(P)}(P), \omega)\varphi(P)(\omega)> \max\limits_{P': \varphi(P)\in P'} \sum\limits_{a, \omega} \hat{f}_{\varphi(P)}(P')(a)u_{S}(\hat{f}_{\varphi(P)}(P'), \omega)\varphi(P)(\omega)
	\end{equation*}
	which is clearly impossible. 
\end{proof}

	\newpage
	\bibliography{references}
	\bibliographystyle{econ-jet}

\end{document}